\documentclass[journal]{IEEEtran}

\usepackage[utf8]{inputenc}
\usepackage{graphicx,epsf}
\usepackage{flushend}

\usepackage{color}

\usepackage[colorlinks,bookmarksopen,bookmarksnumbered,citecolor=red,urlcolor=red]{hyperref}
\usepackage{float}
\usepackage{siunitx}
\usepackage{hhline}
\usepackage{tabularx}
\usepackage{anyfontsize}
\usepackage{placeins}

\newcolumntype{L}[1]{>{\raggedright\arraybackslash}p{#1}}
\newcolumntype{C}[1]{>{\centering\arraybackslash}p{#1}}
\newcolumntype{R}[1]{>{\raggedleft\arraybackslash}p{#1}}

\usepackage{amsmath,amsfonts, amsthm, bm, amssymb,amscd}

\newcommand{\norm}[1]{\left\lVert#1\right\rVert}
\newcommand{\abs}[1]{\left\lvert#1\right\rvert}

\DeclareMathAlphabet{\mathcalON}{OT1}{pzc}{m}{n}

\newcommand \s {{\mathrm s}}
\renewcommand \L {{\mathrm L}}
\newcommand \C {{\mathrm C}}
\newcommand \dd {{\mathrm d}}
\newcommand \ddt {\frac{\dd}{\dd t}}
\DeclareBoldMathCommand \cX{{\cal X}}
\DeclareBoldMathCommand \cA{{\cal A}}
\DeclareBoldMathCommand \cB{{\cal B}}
\newcommand \cP {\mathbf{p}}
\DeclareBoldMathCommand \cQ{{\cal Q}}
\DeclareBoldMathCommand \cC{{\cal C}}
\newcommand \cW {\mathbf{w}}
\DeclareBoldMathCommand \cI{{\cal I}}
\DeclareMathOperator{\e}{\mathrm{e}}
\newcommand \Tlim{\mathcalON{T}}
\newcommand \bfw {\mathbf{w}}
\newcommand \bfc {\mathbf{c}}
\newcommand \bfx {\mathbf{x}}
\newcommand \bfA {\mathbf{A}}
\newcommand \bfB {\mathbf{B}}
\newcommand \bfxhat {\mathbf{\widehat{x}}}
\newcommand \bfchat {\mathbf{\widehat{c}}}
\newcommand \chat {{\widehat{c}}}
\newcommand \xhat {{\widehat{x}}}
\newcommand \OmegaHat {\widehat{\Omega}}
\newcommand \Np {N_{\mathrm p}}
\newcommand \Ns {N_{\mathrm s}}
\newcommand \Nsp {N_{\mathrm s}(\Np+1)}
\newcommand \Ts {T_{\mathrm s}}
\newcommand \fs {f_\mathrm{s}}
\newcommand \iL {i_\mathrm{L}}
\newcommand \vC {v_{\mathrm C}}
\newcommand \vi {v_{\mathrm i}}

\newcommand \Ltwo {{\mathrm{L}^2}}
\newcommand \rma {\mathrm{a}}
\newcommand \rmb {\mathrm{b}}

\newtheorem{theorem}{Theorem}

\newtheorem{remark}{Remark}

\begin{document}
\title{Efficient simulation of DC-DC switch-mode power converters by multirate partial differential equations}
\author{Andreas Pels, Johan Gyselinck, Ruth V. Sabariego, Sebastian Sch\"ops
\thanks{A. Pels and S. Sch\"ops are with the Graduate School of Computational Engineering, Technische Universität Darmstadt, Dolivostra{\ss}e 15, 64293 Darmstadt, Germany, and with the Institut für Theorie Elektromagnetischer Felder, Technische Universität Darmstadt, Schlo{\ss}gartenstra{\ss}e 8, 64289 Darmstadt, Germany.}%
\thanks{J. Gyselinck is with the BEAMS Department, Universit\'e libre de Bruxelles, Avenue Franklin Roosevelt 50, 1050 Bruxelles, Belgium.}%
\thanks{R. V. Sabariego is with the Department of Electrical Engineering, KU Leuven, EnergyVille, Kasteelpark Arenberg 10, 3001 Leuven, Belgium.}}%

\maketitle

\begin{abstract}
  In this paper, Multirate Partial Differential Equations (MPDEs) are used for the efficient simulation of problems with 2-level pulsed excitations as they often occur in power electronics, e.g., DC-DC switch-mode converters. The differential equations describing the problem are reformulated as MPDEs which are solved by a Galerkin approach and time discretization. For the solution expansion two types of basis functions are proposed, namely classical Finite Element (FE) nodal functions and the recently introduced excitation-specific pulse width modulation (PWM) basis functions. The new method is applied to the example of a buck converter. Convergence, accuracy of the solution and computational efficiency of the method are numerically analyzed.
\end{abstract}


\section{Introduction}
Multirate behaviour can be observed in a number of technical applications. In high-frequency electrical circuit simulation, e.g., \cite{Brachtendorf_1996aa,Roychowdhury_2001aa, Mei_2005aa}, the solution often consists of widely separated frequencies with slow and fast varying components. Furthermore, a division of the circuit into subcircuits whose state variables are either latent or active is often possible, especially in highly integrated circuits with many electrical elements \cite{Gunther_1993aa}. The same holds for field-circuit coupled simulations describing the same physical phenomenon, e.g. an electrical circuit coupled to a magnetoquasistatic field model of an electrical machine. In coupled multiphysical simulations, different physical phenomena exhibit different characteristic time constants (e.g. electro-thermal problems) and thus also lead to different rates of variation in the unknowns \cite{Schops_2010aa}.

The solution of the abovementioned problems by conventional time discretization is inefficient as it enforces a step size to resolve the dynamics of the most active components of the system. Thus, the latent parts of the system are resolved with a much smaller time step size than necessary. This results in long simulation intervals and high computational effort. To efficiently solve these problems, various multirate methods have been developed.

Problems described by ordinary differential and differential algebraic equations (ODEs and DAEs) can be split into subsystems \cite{Gunther_1993aa, Striebel_2009aa, Gear_1984aa,Kvaerno_1999aa}, i.e., into several systems of equations describing the latent and active components, respectively. The subsystems are coupled, e.g., by extrapolation and interpolation of the state variables and resolved by different time step sizes and/or methods. A method of this kind for the simulation of power electronics has been proposed by Pekarek et al. \cite{Pekarek_2004aa}. The coupling variables are synchronized in certain intervals. Between the synchronization time instants, the solution of the slow subsystem within the fast subsystem is calculated using a predictor and interpolation. Within the slow subsystem the solution of the fast subsystem is calculated by averaging.

Another recent concept to deal with multirate phenomena is the reformulation of the ODEs or DAEs describing the problem into multirate partial differential equations (MPDEs) \cite{Brachtendorf_1996aa, Roychowdhury_2001aa}. The concept of the MPDEs allows to split the solution into components associated with different explicitly stated time scales $t_1,t_2,\dots,t_m$. MPDEs have already been successfully applied in high-frequency circuit simulation using different solution approaches, e.g., multi-tone harmonic balance \cite{Brachtendorf_1996aa}, multivariate finite difference time domain, hierarchical shooting \cite{Roychowdhury_2001aa} or combinations of different time stepping methods \cite{Mei_2005aa}.

In this paper we focus on the efficient simulation of problems which are excited by periodic 2-level pulsed (control) signals, as is often the case in power electronics, e.g., in DC-DC switch-mode power converters \cite{Mohan_2003aa}.
The system of differential equations describing the application is reformulated into a system of MPDEs. The MPDEs are solved by a combination of a Ritz-Galerkin approach and conventional adaptive time discretization. For the solution expansion of the unknowns, two types of basis functions are proposed: 1. the standard nodal FE basis functions; 2. the excitation-specific pulse width modulation (PWM) basis functions, introduced in \cite{Gyselinck_2013ab}. The latter are designed to represent the ripple component in the output of switch-mode power converters. These are, for the first time, interpreted in the framework of MPDEs and their approximation properties are mathematically analyzed. Since the concept of MDPEs allows to explicitly split the solution into components of different rates and solve them with different methods, it is possible to efficiently simulate not only the steady-state (as done in \cite{Gyselinck_2013ab}), but also the transient behaviour of an application.  
The MPDE approach is validated in the example of a buck converter \cite{Gyselinck_2013ab} in continuous conduction mode, as depicted in Fig. \ref{fig:buckConvSimp}. Its solution, shown in Fig. \ref{fig:solutionBuckConv500Hz}, consists of a fast periodic ripple component and a slowly varying envelope.

The paper is structured as follows. Section~\ref{sec:IntMPDE} introduces the concept of MPDEs as described in the literature and establishes a link between the original system and the MPDEs. Section~\ref{sec:solMPDE} is devoted to the solution of the MPDEs using a combination of a Ritz-Galerkin approach and conventional time discretization. In Section~\ref{sec:BFs} the two different types of basis functions for the solution expansion of the unknowns are presented. Finally in Section~\ref{sec:Numerics} the method is numerically validated on the simplified buck converter and convergence, accuracy and computational efficiency are analyzed. Section~\ref{sec:Conclusion} concludes the work by briefly summarizing the proposed approach and the main results.

\section{Introduction to Multirate Partial Differential Equations}
\label{sec:IntMPDE}
In the following the proposed method is developed starting from a general linear circuit model. Let the vector of $\Ns$ unknown state variables consisting of node voltages and branch currents be given as
\begin{equation}
  {\bfx}(t) =  
  \left[\begin{array}{c} x_1(t) \\ x_2(t) \\ \vdots  \\  x_{\Ns}(t)   \end{array}\right] \, .
\end{equation}
Modified nodal analysis \cite{Ho_1975aa} can be used to determine the system of $\Ns$ first-order linear DAEs or ODEs governing the circuit.
This leads to the initial value problem (IVP)
\begin{equation}
  \begin{aligned}
    \bfA \ddt \bfx(t)+\bfB \bfx(t) &= \bfc(t), \\ 
    \bfx(0)&=\bfx_0,\\
    t \in \Omega&=[0,\Tlim],
  \end{aligned}
  \label{equ:DAEOriginal}
\end{equation}
where $\bfA,\bfB \in \mathbb{R}^{\Ns \times \Ns}$ are matrices, $\bfc(t) \in \mathbb{R}^{\Ns}$ is the vector of excitations, $\bfx_0$ is the vector of initial values and $\Tlim$ determines the simulation interval. 

For $\bfx(t)\in \mathrm{C}^1$, i.e., if $\bfx(t)$ is continuously differentiable, \eqref{equ:DAEOriginal} can be written equivalently as system of MPDEs \cite{Brachtendorf_1996aa, Roychowdhury_2001aa}, introducing $M$ different time scales $t_1,t_2,\dots,t_m$
\begin{equation}
  \begin{aligned}
    \bfA \, &\left(\frac{\partial \bfxhat}{\partial t_1} + \frac{\partial \bfxhat}{\partial t_2} + \dots + \frac{\partial \bfxhat}{\partial t_m}\right) + \bfB \,\bfxhat = \bfchat \, ,\\
    &\hspace{0.3em}(t_1,t_2,\dots,t_m)\in\OmegaHat=[0,\Tlim_1]\times[0,\Tlim_2]\times \dots \times [0,\Tlim_m], \\
  \end{aligned}
  \label{equ:MPDEOriginal}
\end{equation}
where $\bfxhat=\bfxhat(t_1, \dots, t_m)$, $\bfchat=\bfchat(t_1,\dots, t_m)$ are the multivariate forms of $\bfx(t)$, $\bfc(t)$, respectively, and $\Tlim_1, \dots, \Tlim_m$ determine the simulation domains. The vector of state variables $\bfxhat$ is given by
\begin{equation}
  \bfxhat(t_1,\dots,t_m) =  
  \left[\begin{array}{cc} \hat{x}_1(t_1,\dots,t_m) \\ \hat{x}_2(t_1,\dots,t_m) \\ \vdots  \\  \hat{x}_{\Ns}(t_1,\dots,t_m)   \end{array}\right] \, .
\end{equation} 
In the following, a relation between the solution and excitation of \eqref{equ:DAEOriginal} and \eqref{equ:MPDEOriginal} is established, which was first introduced by Brachtendorf et al. \cite{Brachtendorf_1996aa}. They developed a so-called multi-tone harmonic balance method using MPDEs to efficiently simulate high-frequency circuits with more than one fundamental frequency.

Let $\bfxhat(t_1,\dots,t_m) \in \mathrm{C}^1$ be a solution of the MPDEs \eqref{equ:MPDEOriginal} and $\bfchat(t_1,\dots,t_m)$ the corresponding excitation. Then the solution and excitation of the DAEs or ODEs \eqref{equ:DAEOriginal} and MPDEs \eqref{equ:MPDEOriginal} are related by $\bfx(t)=\bfxhat(t+\alpha_1,\dots,t+\alpha_m)$ and $\bfc(t)=\bfchat(t+\alpha_1,\dots,t+\alpha_m)$, respectively, for any fixed $\alpha_1,\dots,\alpha_m \in \mathbb{R}$ \cite{Brachtendorf_1996aa,Knorr_2007aa}.

To proof this statement, the chain rule of differentiation is applied to \eqref{equ:DAEOriginal}, which yields \cite{Brachtendorf_1996aa,Knorr_2007aa}
  \begin{equation}
    \begin{aligned} 
      \left. \bfA \, \ddt \bfx(t) \right|_{t=t_0} &= \left. \bfA \, \ddt \bfxhat(t+\alpha_1,\dots,t+\alpha_m)\right|_{t=t_0}\\
      &= \bfA \, \bigg[\frac{\partial \bfxhat(t_1,\dots,t_m)}{\partial t_1}+ \dots \\
      &\hspace{1.5em}+ \left.\frac{\partial \bfxhat(t_1,\dots,t_m)}{\partial t_m}\bigg] \right|_{t_1=t_0+\alpha_1,\dots,t_m=t_0+\alpha_m}  \\
      &\hspace{-0.05em}\overset{\eqref{equ:MPDEOriginal}}{=} \, \bfchat(t_0+\alpha_1,\dots,t_0+\alpha_m)\\
      &\hspace{1.2em}-\bfB \, \bfxhat(t_0+\alpha_1,\dots, t_0+\alpha_m)\\
      &=\bfc(t_0)-\bfB \, \bfx(t_0).
    \end{aligned}
  \end{equation}
  
Thus if a solution of the MPDEs~\eqref{equ:MPDEOriginal} can be found for a multivariate right-hand side fulfilling $\bfc(t)=\bfchat(t+\alpha_1,\dots,t+\alpha_m)$, the solution of the DAEs or ODEs~\eqref{equ:DAEOriginal} can be extracted from the multivariate solution using $\bfx(t)=\bfxhat(t+\alpha_1,\dots,t+\alpha_m)$.
To solve the MPDEs \eqref{equ:MPDEOriginal}, initial and boundary conditions have to be imposed. 
As only IVPs are considered, whose solution can be separated into periodic and non-periodic parts, the setting of envelope-modulated solutions \cite{Roychowdhury_2001aa} is appropriate.
We therefore define initial and boundary conditions to the MPDEs \eqref{equ:MPDEOriginal} as 
\begin{equation}
  \begin{aligned}
    \bfxhat(t_1,t_2+T_2,\dots,t_m+T_m)&=\bfxhat(t_1,t_2,\dots,t_m)\\
    \bfxhat(0,t_2,\dots,t_m)&=h(t_2,\dots,t_m),
  \end{aligned}
\end{equation}
where $h(t_2,\dots,t_m)$ is a function specifying the initial conditions and $T_2,\dots,T_m$ are time intervals of periodicity. 

For the sake of simplicity, hereafter, we restrict $m$ to two time scales ($m=2$) leading to the mixed initial boundary value problem
\begin{equation}
  \begin{aligned}
    \bfA \, \left(\frac{\partial \bfxhat(t_1, t_2)}{\partial t_1} + \frac{\partial \bfxhat(t_1, t_2)}{\partial t_2}\right) + \bfB \bfxhat(t_1, t_2) &= \bfchat(t_1, t_2) \, ,\\
    \bfxhat(t_1,t_2+T_2)&=\bfxhat(t_1,t_2), \\
    \bfxhat(0,t_2+T_2)&=\bfxhat_0(t_2).
  \end{aligned}
  \label{equ:MPDEOriginalTwoTS} 
\end{equation}

In the following section, we will apply the MPDE framework to problems with discontinuous right-hand sides. Existence and uniqueness can still be assured by the Carath\'eodory conditions. The interested reader is referred to, e.g., \cite{Flippov_1988aa}. However, in the case of the pulsed excitations introduced in the next section, a piecewise analysis is possible and a detailed discussion is not needed.

\section{Solution of the Multirate Partial Differential Equations}
\label{sec:solMPDE}
In this section we focus on the solution of the MPDEs for applications with PWM (pulsed) excitation. The switching cycle $\Ts$ and duty cycle $D$ are assumed to be constant. 
We propose the following procedure for solving: 1.) a Ritz-Galerkin approach is applied to one dimension of the MPDEs \eqref{equ:MPDEOriginalTwoTS}; 2.) the remaining linear system of ODEs or DAEs is solved with conventional time discretization. 

\subsection{Solution expansion by basis functions}
In a first step the multivariate solution $\bfxhat(t_1,t_2)$ is expanded into a finite set of basis functions and coefficients. It reads
\begin{equation}
  \xhat_j^h(t_1, t_2) := \sum_{k=0}^{\Np} p_k(t_2)\, w_{j,k}(t_1) \, ,
  \label{equ:solExp}
\end{equation}
where $\xhat_j^h$, $1\leq j \leq \Ns$ is the $j$-th approximated state variable, $p_k(t_2)$, $0 \leq k \leq \Np$ are periodic basis functions and $w_{j,k}(t_1)$ are coefficients. The superscript $h$ in $\xhat_j^h(t_1, t_2)$ denotes that it is an approximation to $\xhat_j(t_1, t_2)$. By defining the expansion as above we associate the slowly varying envelope with the time scale $t_1$, which will be therefore referred to as the slow time scale, and the fast periodically varying ripples with the time scale $t_2$, which will be referred to as fast time scale.
The basis functions are periodic $p_k(t_2)=p_k(t_2+\Ts)$ with switching cycle $\Ts$, which can be accounted for by introducing the relative time $\tau \in [0,1]$
\begin{equation}
  \tau \ = \ \frac{t_2}{T_\s} \ \text{modulo} \ 1 \, .
\end{equation}
The switching cycle $\Ts$ is related to the switching frequency by $\Ts=\frac{1}{\fs}$. For simplicity the basis functions will in the following be expressed as functions of the relative time~$p_k(\tau)$.

Inserting the solution expansion into the partial derivatives from \eqref{equ:MPDEOriginalTwoTS} yields
\begin{equation}
  \frac{\partial \xhat_j^h}{\partial t_1} \, = \sum_{k=0}^{\Np}  \Big(p_k(\tau)\, \frac{\dd w_{j,k}}{\dd t_1} 
  \Big) \, ,
  \label{equ:partDeriv_t1}
\end{equation}
\begin{equation}
  \frac{\partial \xhat_j^h}{\partial t_2} \, = \, \sum_{k=0}^{\Np}  \Big( 
  \frac{\dd p_k}{\dd \tau} \, \frac{\dd \tau}{ \dd t_2} \, w_{j,k}(t_1)\Big)
  \label{equ:partDeriv_t2}
\end{equation}
with
\begin{equation}
  \frac{\dd \tau}{\dd t_2} =  \frac{1}{ T_\s} = f_\s \, .
\end{equation}
The solution expansion in matrix form is
\begin{equation}
  \xhat_j^h(t_1,t_2) = \cP ^\top\!(\tau) \,\cW_j(t_1)  \, ,
  \label{equ:expMat}
\end{equation}
where $\cP$ and $\cW_j$ are column vectors of length $\Np+1$
\begin{equation}
  \cP(\tau) =  
  \left[\begin{array}{c} p_0 \\ p_1(\tau) \\ p_2(\tau) \\ \vdots \\  p_{N_{\mathrm p}}(\tau)   \end{array}\right] \ , \ \
  \quad
  \cW_j(t) \ = \ 
  \left[\begin{array}{c} w_{j,0}(t) \\ w_{j,1}(t) \\ w_{j,2}(t) \\ \vdots \\  w_{j,\Np}(t)   \end{array}\right]
  \, .
\end{equation}
The sum of the partial derivatives \eqref{equ:partDeriv_t1} and \eqref{equ:partDeriv_t2} can finally be written as 
\begin{equation}
  \frac{\partial \xhat_j^h}{\partial t_1} + \frac{\partial \xhat_j^h}{\partial t_2}\, = \cP^\top\!(\tau) \, \frac{\dd \cW_j}{\dd t_1} \, + \, f_\s \, \frac{\dd \cP^\top}{\dd\tau} \, \cW_j(t_1).
  \label{equ:sumPartDeriv}
\end{equation}

\subsection{Galerkin approach}
The Ritz-Galerkin approach is applied to the MPDEs with respect to the fast time scale $t_2$ in the interval~$[0,\Ts]$
\begin{equation}
  \begin{aligned}
    \int\limits_{0}^{T_\s} \Bigg( &{\bfA} \, \left(\frac{\partial \bfxhat^h}{\partial t_1} + \frac{\partial \bfxhat^h}{\partial t_2}\right) + \bfB \, \bfxhat^h(t_1,t_2) \\
    -&\bfchat(t_1,t_2) \Bigg) \, p_l(\tau(t_2)) \, \dd t_2 \ =             \ 0 \, , \, \forall l=0,\dots,\Np,
  \end{aligned}
  \label{equ:gal}
\end{equation}
i.e., the MPDEs are weighted by the same basis functions used for the solution expansion. 

Let the matrices $\cI$ and $\cQ$ be given as 
\begin{equation}
  \cI  \ = \ \Ts \int\limits_0^1  \cP(\tau) \, \cP^\top \!(\tau)  \,\dd \tau \, , \quad \cQ \ = \ - 
  \int\limits_0^1  \frac{\dd \cP}{\dd \tau } \, \cP^\top(\tau)  \,  \dd \tau  \, .
\end{equation}
Inserting the relation \eqref{equ:sumPartDeriv} into \eqref{equ:gal} leads to
\begin{equation}
  \cA\, \frac{\dd  \cW}{\dd t_1} + \cB \, \cW(t_1) \ = \ \cC(t_1) \, ,
  \label{equ:ReducedMPDESystem}
\end{equation}
where 
\begin{equation}
  \cW(t_1) = 
  \left[\begin{array}{cccc}
    \cW_{1}(t_1)    \\ 
    \cW_{2}(t_1)     \\ 
    \vdots   \\
    \cW_{\Ns}(t_1)
  \end{array}\right] \, 
\end{equation}
is the unknown vector of $\Nsp$ coefficients and $\cA, \cB \in \mathbb{R}^{\Nsp\times\Nsp}$ and $\cC\in \mathbb{R}^{\Nsp}$ are, using the Kronecker product, given by
\begin{align}
  \cA&=\bfA\otimes\cI,\\
  \cB&=\bfB\otimes\cI+\bfA\otimes\cQ,\\
  \cC(t_1)&=\int\limits_{0}^{T_\s}
  \left[\begin{array}{cccc}
    \chat_{1}(t_1,t_2)  \,\cP(\tau(t_2))  \\ 
    \vdots   \\
    \chat_{\Ns}(t_1,t_2)  \,\cP(\tau(t_2))  \\ 
  \end{array}\right]
  \dd t_2 \, .
  \label{equ:integralC}
\end{align}
Note that the sparsity pattern of the matrices $\cI$ and $\cQ$ depends on the choice of the basis functions. 

\subsection{Time discretization}
The equations \eqref{equ:ReducedMPDESystem} are now formulated only in $t_1$. According to \eqref{equ:integralC} their right-hand side naturally depends on the right-hand side of the MPDEs, which only needs to satisfy the relation $\bfc(t)=\bfchat(t,t)$ according to Section \ref{sec:IntMPDE}. As a result infinitely many choices for $\bfchat(t_1, t_2)$ are possible. However to minimize the dynamic of the system \eqref{equ:ReducedMPDESystem} and thus maximizing the efficiency of the approach, it is reasonable to head for a constant right-hand side. As $\bfc(t)$ is periodic with switching cycle $\Ts$ for the considered problems, we choose $\bfchat(t_1,t_2)=\bfc(t_2)$. Inserting this into \eqref{equ:integralC}, the time scales $t_1$ and $t_2$ vanish which leads to
\begin{equation}
  \cC(t_1)=\mathit{constant}\,. 
\end{equation}
Note that the system \eqref{equ:ReducedMPDESystem} is $\Np+1$ times larger than the original one \eqref{equ:DAEOriginal}.

\section{Choice of basis functions for solution expansion}
\label{sec:BFs}
We propose the use of standard FE nodal functions as in classical finite element methods (FEM) or the PWM basis functions introduced in \cite{Gyselinck_2013ab}. 
\subsection{Finite element nodal basis functions}
To start with the FE nodal functions of first order, let us introduce a division of the relative time interval $[0,1]$ into elements, such that $0 < \tau_1 < \tau_2 < \dots < \tau_{\Np} < 1$, where the $\tau_k$ are the nodes defining the elements. 
The nodal basis functions are piecewise linear functions defined as
\begin{equation}
  p_k(\tau)=\left\{
  \begin{array}{ll}
    \frac{\tau-\tau_{k-1}}{\tau_k-\tau_{k-1}} & \text{for }\tau \in (\tau_{k-1}, \tau_k] \\
    \vspace{0.1em} \\
    \frac{\tau-\tau_{k+1}}{\tau_k-\tau_{k+1}} & \text{for }\tau \in (\tau_k, \tau_{k+1}) \\
    \vspace{0.1em} \\
    0 & \text{otherwise}
  \end{array}
  \right. .
\end{equation}
To enforce periodicity on the interval $[0,1]$ we set the basis functions at the boundary, i.e., $p_1(\tau)$ and $p_{\Np}(\tau)$, to zero
\begin{align}
  p_1(\tau)&=
  \begin{array}{ll}
    0 & \text{for } \tau \in [0,1] \\
  \end{array}
  \\
  p_{\Np}(\tau)&=
  \begin{array}{ll}
    0 & \text{for }\tau \in [0,1] \\
  \end{array}.
\end{align}
To resolve the envelope, we introduce an additional constant basis function
\begin{equation}
  p_0(\tau)=1 \quad\, \text{for } \tau \in [0,1],
\end{equation}
The FE nodal basis as defined above is depicted in Fig.~\ref{fig:feBasis}.
As the FE nodal functions offer local support, except $p_0(\tau)$, the matrices $\cI, \cQ$ are sparsely populated matrices. Due to the constant basis function $p_0(\tau)$ supporting the entire relative time interval $[0,1]$, the matrices are not purely banded matrices as in classical FE methods.

Instead of setting the boundary functions to zero and defining an additional constant basis function, it is also possible to enforce periodic boundary conditions on the set of standard FE nodal functions in the final system of equations.
\begin{figure}
    \centering
    \includegraphics{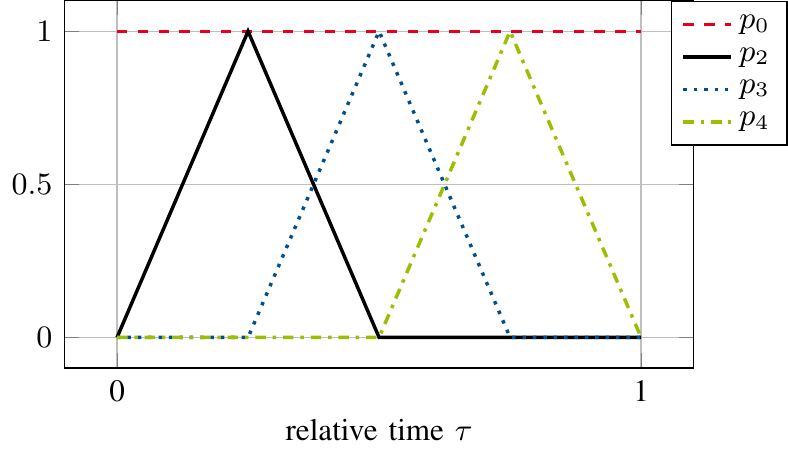}
    \caption{FE nodal basis functions $p_k(\tau)$, $k\in\{0,2,3,4\}$.}
    \label{fig:feBasis}
\end{figure}

\begin{figure}
    \centering
    \includegraphics{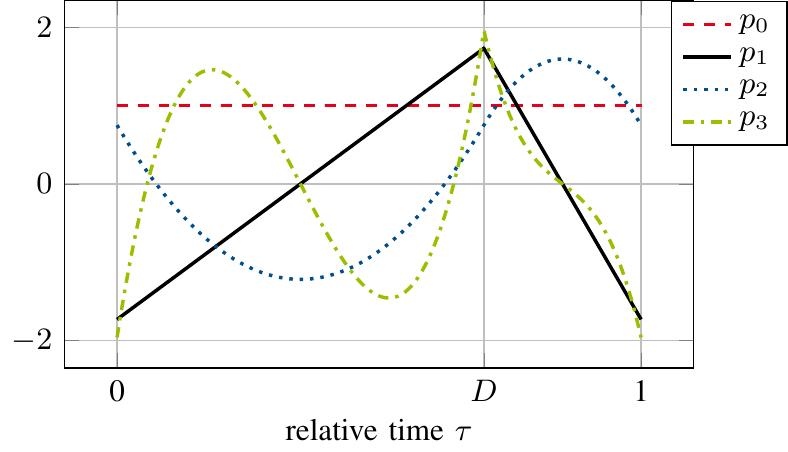}
    \caption{PWM basis functions $p_k(\tau)$, $k\in \{0,1,2,3\}$.}
    \label{fig:pwmBasis}
\end{figure}

\subsection{PWM basis functions}
A problem-specific choice of basis functions in case of a-priori known duty cycle is the PWM basis functions, which were developed in \cite{Gyselinck_2013ab}. The a-priori knowledge enables us to build the basis functions such that they mimic the shape of the ripple components in the solution by construction. The zero-th basis function is $p_0(\tau)=1$ which resolves the envelope as in the case of nodal basis functions. 
The PWM basis is iteratively built starting from the normalized, zero average, piecewise linear basis function $p_1(\tau)$ defined as \cite{Gyselinck_2013ab}
\begin{equation}
  \renewcommand*{\arraystretch}{1.3}
  \displaystyle
  p_1(\tau) \ = \ \left\{  
  \begin{array}{ll} 
    \sqrt{3} \ \frac{2\tau-D}{D} \quad & \mathrm{if} \ \ \ 0 \leq \tau \leq D \\ 
    \sqrt{3} \ \frac{1+D-2\tau}{1-D}  \quad & \mathrm{if} \  \ \ D \leq \tau \leq 1    \end{array}
  \right. .
  \label{equ:PWMbasis_p1}
\end{equation}
The higher-order basis functions $p_k(\tau)$, $2\leq k \leq \Np$ are obtained recursively by integrating the basis functions of lower order $p_{k-1}(\tau)$ ensuring $C^0$-continuity
\begin{equation}
  p_k^\star(\tau) = \int_D^\tau p_{k-1}(\tau') \,\dd\tau'  \, .
\end{equation}
This extended set of basis functions is successively orthonormalized, starting from $k=2$, by orthogonalizing
\begin{equation}
  \overline{p}_k(\tau) = p_k^\star(\tau) - \sum_{l=0}^{k-1} p_l(\tau) \, \int\limits_0^1 p_l(\tau) p_k^\star(\tau) \dd \tau  \, ,
\end{equation}
and normalizing
\begin{equation}
  p_k(\tau)= \frac{\overline{p}_k(\tau)}{\sqrt{\int\limits_0^1 \overline{p}_k(\tau) \overline{p}_k(\tau) \dd \tau}},
  \label{equ:PWMbasis_normalizing}
\end{equation}
which corresponds to a Gram-Schmidt orthonormalization \cite{Trefethen_1997aa}. Note that it is possible to calculate the PWM basis functions analytically. 
The basis functions of order up to 3 are depicted in Fig. \ref{fig:pwmBasis}. 

Opposed to the FE nodal functions, the PWM basis functions are global polynomials on the relative time interval $[0,1]$ as in spectral methods and offer the same accuracy with less degrees of freedom compared to the nodal basis functions \cite{Canuto_2006aa}. 
Due to the orthonormality of the PWM basis functions the matrix $\cI$ is the identity matrix. The matrix $\cQ$ is dense, however only 25\% are non-zero elements.

The approximation properties of these specific basis functions have been studied up to now only numerically \cite{Gyselinck_2013ab}. The basis functions are by construction restricted to represent piecewise exponential solutions. Their properties are studied analytically in the Appendix for a duty cycle of $D=0.5$.

\section{Numerical results}
\label{sec:Numerics}
In this section the method is numerically validated. Computational efficiency, accuracy and convergence results are presented. All calculations have been performed in GNU Octave \cite{Eaton_2015aa}. For solving equation \eqref{equ:ReducedMPDESystem}, an implicit Runge-Kutta method of order 5 with 6 stages is used. 
For step size prediction the estimated error is measured in the infinity norm instead of the 2-norm as originally proposed in \cite{Hairer_1996aa}, p. 124, i.e., 
\begin{equation}
  \norm{\mathit{err}} = \max\limits_{i=1,\dots,N}\abs{\frac{\mathit{err}_i}{\mathit{sc}_i}},
\end{equation}
where $N$ is the dimension of the equation system and $\mathit{err}_i$ is the estimated error of the $i$-th solution component in each step. The quantity $\mathit{sc}_i$ depends on the relative and absolute tolerance. For more information the reader is referred to \cite{Hairer_1996aa}.
The absolute tolerance is fixed at $\mathit{abstol}=10^{-10}$ so that the error estimation is controlled by the relative tolerance $\mathit{reltol}$. The solver supports dense output which is used in reconstructing the MPDE solution.

\subsection{Test case}
\label{sec:testCase}
The test case is a buck converter circuit \cite{Gyselinck_2013ab} as depicted in Fig. \ref{fig:buckConvComplete}. The buck converter consists of a DC voltage source $V_\mathrm{i}$, a switch (e.g. an IGBT), a diode, an inductor (consisting of inductance $L$ and resistance $R_\mathrm{L}$) and a capacitor (capacitance $C$). At the output a load resistance $R$ is connected. The switch is controlled by a 2-level pulsed signal, which closes and opens the switch at switching frequency $\fs$ and with a duty cycle $D$. 

\begin{figure}
  \centering
  \includegraphics{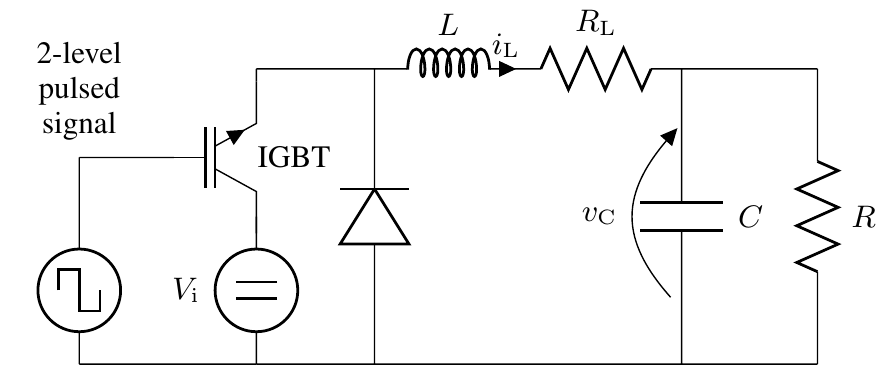}
  \caption{Circuit of a step-down buck converter.}
  \label{fig:buckConvComplete}
\end{figure}

Assuming continuous conduction mode ($\iL>0$), an ideal switch, and an ideal diode, the buck converter can be simplified as depicted in Fig. \ref{fig:buckConvSimp}\cite{Mohan_2003aa}. The switch and diode have been removed and the voltage source has been replaced by a pulsed voltage source $\vi(t)$, which output voltage alternates between $v_{\mathrm{i},\mathrm{off}}=0\,\mathrm{V}$ and $v_{\mathrm{i},\mathrm{on}}=V_\mathrm{i}$, i.e. 
\begin{equation}
  v_\mathrm{i}(t)= 
  \left\{\begin{array}{ll}
    V_\mathrm{i} & \text{for } \tau(t) \in [0,D]  \\
    0 & \text{otherwise}
  \end{array}\right..
  \label{equ:excitationBuckConv}
\end{equation}
\begin{figure}
  \centering
  \includegraphics{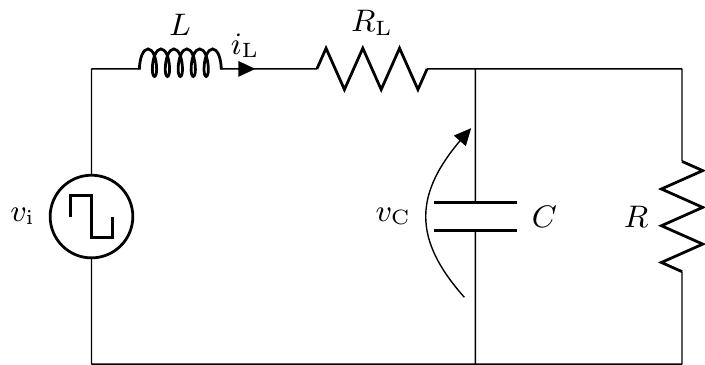}
  \caption{Simplified circuit of the buck converter in continuous conduction mode, with ideal switch and ideal diode.}
  \label{fig:buckConvSimp}
\end{figure}%
The circuit can be described by two state variables, namely the current through the coil $\iL(t)$ and the voltage across the capacitor $\vC(t)$, which is also the output voltage of the buck converter. Using Kirchhoff's circuit laws leads to the first-order linear ODEs
\begin{equation}
  {\left[\begin{array}{cc} L & 0  \\ 0 & C  \end{array}\right]}_{}
  \frac{\dd}{\dd t}{\left[\begin{array}{cc} i_\L \\ v_\C \end{array}\right]}_{}
  + {\left[\begin{array}{cc} R_\L & 1  \\ -1 & 1/R  \end{array}\right]}_{}  
  \!
  {\left[\begin{array}{cc} i_\L \\ v_\C \end{array}\right]}_{}
  \,=\,
  { 
    \left[\begin{array}{cc} v_{\mathrm i}(t) \\ 0  \end{array}\right]}_{} \, ,
  \label{equ:ode_buck}
\end{equation}
where the following parameter values are chosen:
\begin{itemize}
  \item $V_{\mathrm i} =100\,$V;
  \item $f_\s=500$\,Hz;
  \item $D=0.7$;
  \item $L=1\,$mH, $R_{\mathrm L}=10$\,m$\Omega$;
  \item $C=100$\,$\mu$F;
  \item $R=0.8\,\Omega$.
\end{itemize}
The initial conditions are set to $\vC(0)=0$ and $\iL(0)=0$.

As reference solution a closed-form analytic solution of the buck converter ODE \eqref{equ:ode_buck} is used. 
Figs. \ref{fig:solutionBuckConv500Hz} and \ref{fig:solutionBuckConv5000Hz} show the voltage at the capacitor and current through the coil of the buck converter for $\fs=500\,$Hz and $\fs=5000\,$Hz, respectively.
The solution consists of a slowly varying envelope and ripple components which are periodic. Increasing the switching frequency, the magnitude of the ripples decreases.

\begin{figure}
  \centering
  \includegraphics{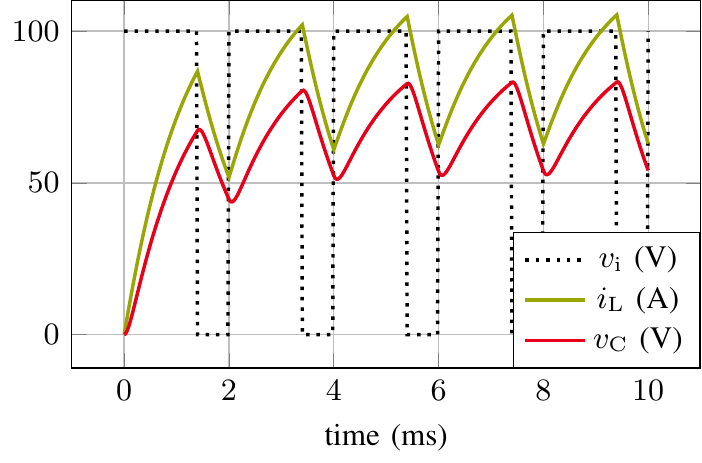}
  \caption{Solution of the buck converter for switching frequency $\fs=500\,$Hz~/~switching cycle $\Ts=2\,$ms and fixed duty cycle $D=0.7$.}
  \label{fig:solutionBuckConv500Hz}
\end{figure}
\begin{figure}
  \centering
  \includegraphics{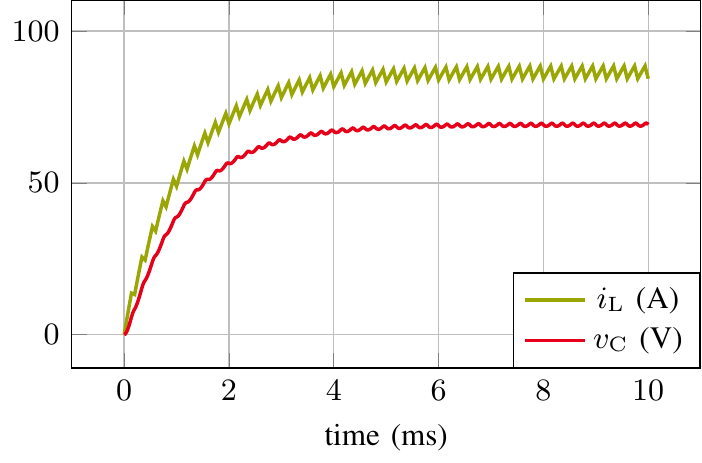}
  \caption{Solution of the buck converter for switching frequency $\fs=5000\,$Hz~/~switching cycle $\Ts=0.2\,$ms and fixed duty cycle $D=0.7$.}
  \label{fig:solutionBuckConv5000Hz}
\end{figure}

\subsection{Multirate solution}
To obtain the multirate solution of the buck converter circuit the MPDE approach as described before is applied. For the FE nodal functions equidistant spacing between the nodes dividing the relative time interval $[0,1]$ into elements is used. The number of basis functions is always chosen such that the jump of the excitation as defined in \eqref{equ:excitationBuckConv} occurs at a time instant which coincides with a node. Thereby the $C^0$ continuity in the solution coincides exactly with a node and is properly represented. For a duty cycle of $D=0.7$ this corresponds to $\Np \in \{11,21,31,41,...\}$. For the PWM basis functions no special care is needed to choose $\Np$ as they take the duty cycle $D$ into account by construction.

The equation system \eqref{equ:ReducedMPDESystem} is solved for the vector of coefficients $\bfw(t_1)$. To find the initial values $\bfw(0)$, the steady-state solution of the system is calculated
\begin{equation}
  \cB\bfw(0)=\cC(0).
\end{equation}
The coefficients $w_{j,0}$ corresponding to the constant basis function $p_0(t_2)$ are set as such that the solution satisfies the initial condition $v_\mathrm{C}(0)=0$ and $i_\mathrm{L}(0)=0$. 

The coefficients for 12 basis functions (11 FE nodal functions + 1 constant function), i.e., $\Np=11$, are exemplary depicted in Fig. \ref{fig:MPDEcoefficients} for the capacitor voltage after solving. All coefficients except $w_{j,0}$ stay constant during the simulation time, i.e. the coefficients controlling the shape of the ripples do not change. 

\begin{figure}[h!]
  \centering
  \includegraphics{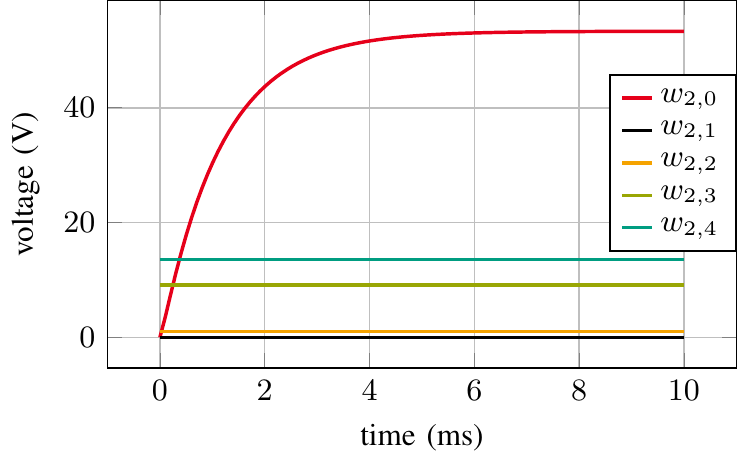}
  \caption{First five coefficients of the solution expansion of the capacitor voltage versus the time, calculated by solving the ODE~\eqref{equ:ReducedMPDESystem} for the buck converter using FE nodal functions ($\Np=11$).}
  \label{fig:MPDEcoefficients}
\end{figure}
\begin{figure}[h!]
  \centering
  \includegraphics{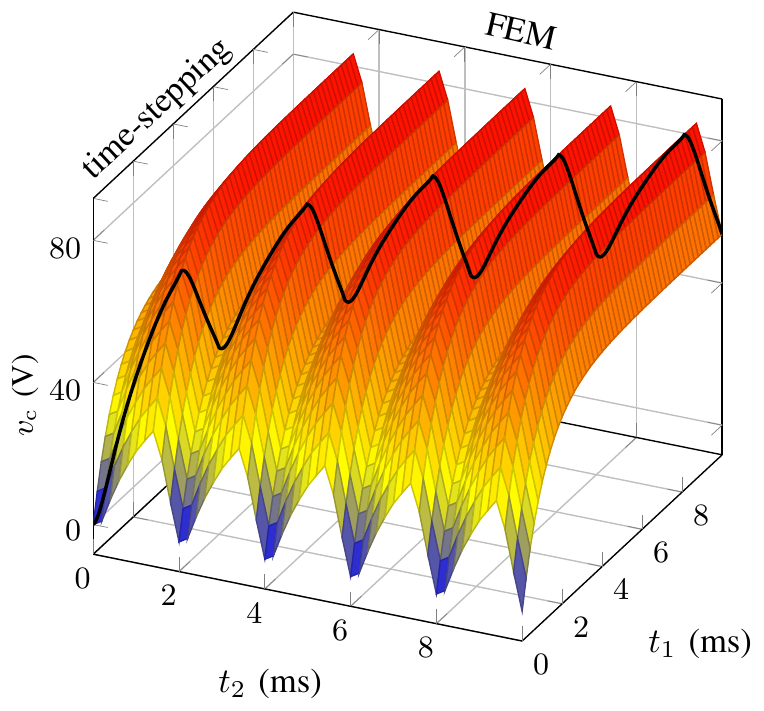}
  \caption{Multivariate solution $\bfxhat_2(t_1,t_2)$ (PWM basis functions) of the buck converter at $f_\mathrm{s}=500\,$Hz. Solution of original ODE marked in black.}
  \label{fig:MPDEHatFunc3D}
\end{figure}

The multivariate solution is reconstructed using the solution expansion~\eqref{equ:solExp}. As the solver often uses less time steps than for the original equations \eqref{equ:ode_buck}, it is taken advantage of dense output to extract a reasonably fine sampled solution. Fig.~\ref{fig:MPDEHatFunc3D} shows the result $\bfxhat(t_1,t_2)$ in a 3D plot. Along the time axis $t_1$ the slow dynamic resolved by time discretization can be observed while along the time axis $t_2$ the high dynamic resolved by the Galerkin approach is visible. The solution of the original equations \eqref{equ:ode_buck} is marked as black line and can be extracted using $\bfx(t)=\bfxhat(t,t)$ according to Section \ref{sec:IntMPDE}.

\subsection{Convergence}
To compare the two types of basis functions used for the solution expansion, the convergence of the solution with respect to $\Np$ and the tolerance of the solver is examined. We consider the simulation time interval $\Omega=[0,10]\,$ms. As reference solution for the buck converter, a closed-form analytic solution is calculated. 
The following analysis is restricted to the output voltage of the buck converter, i.e., the voltage at the capacitor. The convergence behaviour of the current through the inductor is similar.
Let $t\in\Omega$ and define the relative $\Ltwo$-error of the solution by  
\begin{equation}
  \epsilon(\mathit{reltol}, n) = \frac{||v_\mathrm{C,ref}(t)-v_\mathrm{C}^h(\mathit{reltol},n,t)||_{\Ltwo(\Omega)}}{||v_\mathrm{C,ref}(t)||_{\Ltwo(\Omega)}}
  \label{equ:l2erroreps}
\end{equation}
where $v_\mathrm{C}^h(\mathit{reltol},n,t)$ is the voltage at the capacitor calculated by the MPDE approach for different relative tolerance, number of basis functions and time instants and $v_\mathrm{C,ref}(t)$ is the respective reference solution. 
The $\Ltwo$-norm is approximated by numerical quadrature using the mid-point rule. In the following $\epsilon(\mathit{reltol}, n)$ will simply be referred to as error.

The error is evaluated at a fixed number of 500 samples per period $\Ts$. For this, again, the dense output feature of the solver is used.

Fig. \ref{fig:convergence} shows the convergence of the error $\epsilon(\mathit{reltol},n)$ using nodal basis functions (nodal BFs) with $h$-refinement and the PWM basis functions (PWM BFs) with $p$-refinement for a fixed relative tolerance of $\mathit{reltol}=10^{-6}$ for the time stepper. It's tolerance $\mathit{reltol}$ determines a limit for the accuracy of the solution. 
To ensure that the employed tolerance is small enough, the error $\epsilon$ is compared for $\mathit{reltol}=10^{-6}$ and $\mathit{reltol}=10^{-8}$. The absolute difference between the errors for a maximum number of basis functions, $\Np=12$ for the PWM basis, and $\Np=131$ for the FE basis, is several orders of magnitude smaller than the obtained error $\epsilon$. A relative tolerance of $\mathit{reltol}=10^{-6}$ is therefore adequate for all calculations.

According to Fig. \ref{fig:convergence}, the solution using PWM basis functions converges significantly faster than the solution using nodal basis functions with respect to the number of basis functions $\Np$.

\begin{figure}[t]
  \centering
  \includegraphics{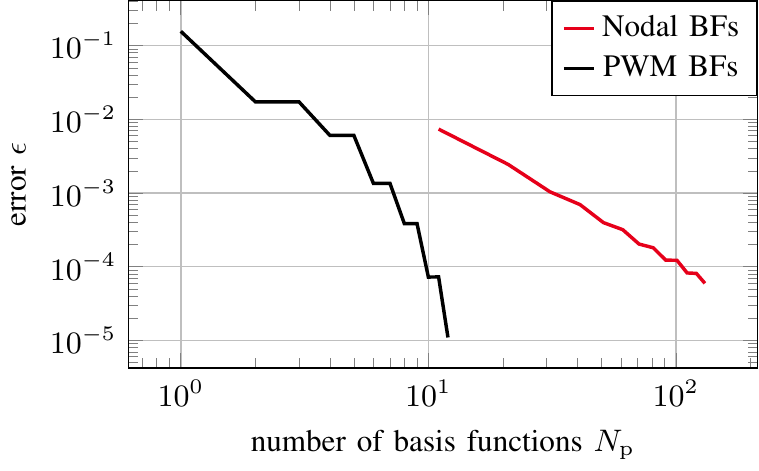}
  \caption{Error $\epsilon$ versus the number of basis functions $\Np$ for nodal basis functions ($h$-refinement) and PWM basis functions ($p$-refinement). The method converges with both types of basis functions. PWM basis functions show higher convergence rate.}
  \label{fig:convergence}
\end{figure}

\subsection{Computational efficiency}
To validate the efficiency of the method, the MPDE approach with nodal and PWM basis functions is compared to classical time discretization of the original ODEs \eqref{equ:ode_buck}. For the time discretization the accuracy is controlled by varying the relative tolerance of the solver while for the MPDE approach we fix the relative tolerance (at $\mathit{reltol}=10^{-6}$), and vary the number of basis functions $\Np$ to achieve a certain accuracy. 
Fig.~\ref{fig:efficiency} shows that the efficiency in terms of time for solving the differential equation systems of the MPDE approach depends on the choice and number of basis functions. While the PWM basis functions yield excellent efficiency, the MPDE approach using nodal functions becomes inferior than time discretization for about $\Np=71$. To better understand this effect, two additional quantities are examined. 

Fig.~\ref{fig:numFunEval} shows the error versus number of function evaluations. For the time discretization this number increases to reach higher accuracy as more time steps are necessary. For the MPDE approach, due to the slow dynamics of the equation system \eqref{equ:ReducedMPDESystem}, much less time steps and thus less function evaluations are needed. For higher accuracy (i.e. increasing $\Np$) the number of function evaluations even decreases. This effect results from adding additional basis functions by which there is more a-priori information on the solution already taken into account. The envelope stored in the zero-th coefficient $w_{j,0}$ including its initial value to ensure the initial conditions of the buck converter therefore varies with different $\Np$ and the ODE solver needs less time steps and thus less function evaluations.
In Fig.~\ref{fig:averTimePerFunEval} the error versus the average time per function evaluation is shown. In this plot the effect of larger equation system in the MPDE approach becomes visible. The average time increases dramatically for the MPDE approach with nodal basis functions as the number of basis functions $\Np \in \{11,\dots,131\}$ is large while for the PWM basis functions the effect is much smaller due to smaller $\Np \in \{1,\dots,12\}$. For conventional time discretization the average time per function evaluation is constant as the size of the equation system does not change and therefore the computational effort per step is constant. The effects visible in Figs.~\ref{fig:numFunEval}~and~\ref{fig:averTimePerFunEval} determine the overall efficiency depicted in Fig.~\ref{fig:efficiency}. In conclusion, for the FE nodal functions this means that the effect of increasing size of equation systems and therefore more effort per step begins to outweigh the advantage of less required time steps for $\Np=71$ and larger. 

The reconstruction of the solution using the solution expansion \eqref{equ:solExp} is not taken into account in the above efficiency measurements. The time for evaluation depends mainly on the number of samples at which the solution is reconstructed. If the number of samples per period is known, the evaluation of the basis functions can be done a-priori. As a result, the reconstruction of the solution is cheap. In the case of 500 samples per period it takes considerably less than $1\,$ms and can therefore be neglected.

Note that the speedup of the MPDE approach compared to time discretization can be expected to increase if larger time intervals are considered or higher switching frequencies $f_\mathrm{s}$ are used. The higher the frequency, the more ripples have to be resolved. Time discretization therefore needs more and more time steps in the same time interval while for the MPDE approach the number of time steps do not change as the periodically varying ripples are resolved by the Galerkin approach. The same happens for increasing time intervals and fixed switching frequency.

\begin{figure}
  \centering
  \includegraphics{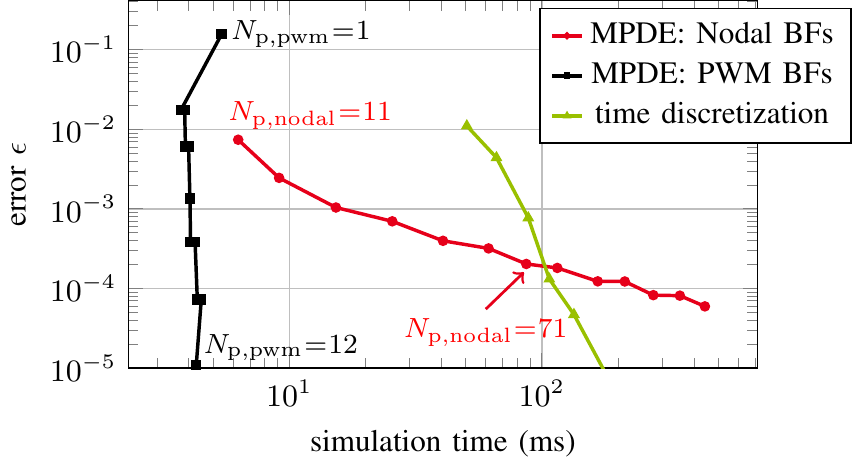}
  \caption{Error $\epsilon$ versus the simulation time in terms of solving the differential equation systems. The MPDE approach with nodal and PWM basis functions is faster than conventional time discretization for a small number of basis functions yielding the same accuracy.}
  \label{fig:efficiency}
\end{figure}
\begin{figure}
  \centering
  \includegraphics{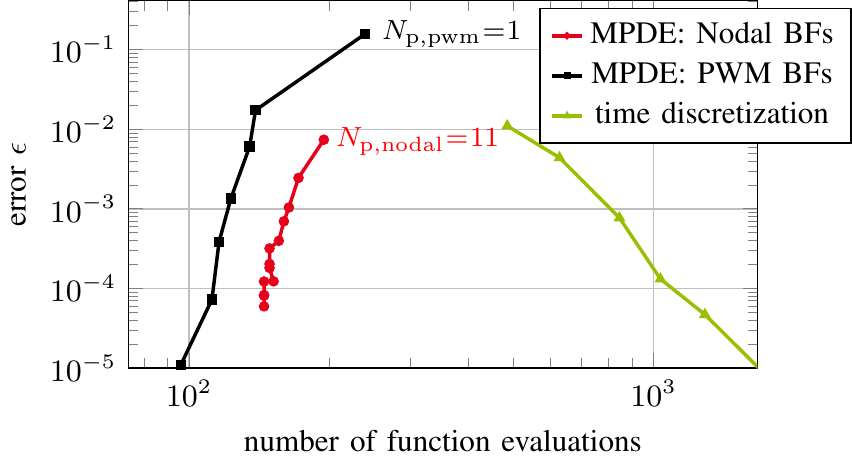}
  \caption{Error $\epsilon$ versus the number of function evaluations. For the MPDE approach the number of function evaluations is significantly smaller than for conventional time discretization to obtain the same accuracy.}
  \label{fig:numFunEval}
\end{figure}
\begin{figure}
  \centering
  \includegraphics{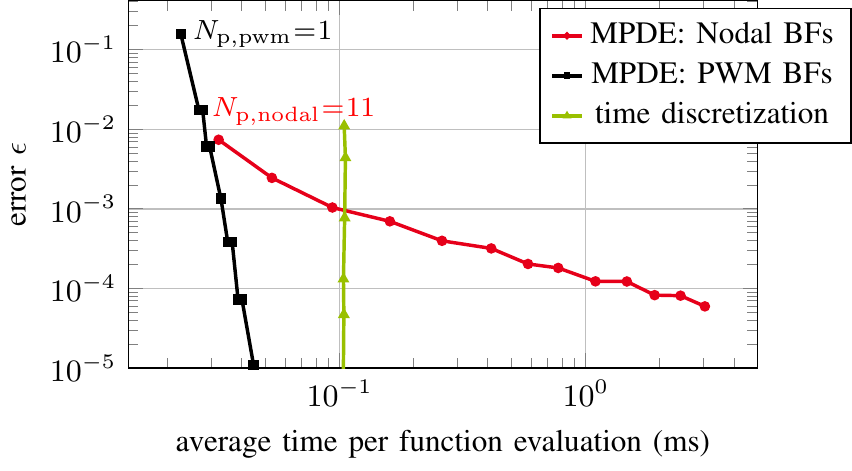}
  \caption{Error $\epsilon$ versus the average time per function evaluations. While for conventional time discretization the average time per function evaluation is constant, it increases for the MPDE approach due to the larger equation systems.}
  \label{fig:averTimePerFunEval}
\end{figure}

\section{Conclusion}
\label{sec:Conclusion}
An efficient and accurate approach to simulate PWM driven applications with constant switching and duty cycle has been presented. The linear circuit model of the power converter is first reformulated as MPDEs. A Galerkin approach and time discretization are used to solve the MPDEs. By this the fast periodically varying components of the solution are taken into account by the basis functions and the time discretization only resolves the dynamics of the envelope. This leads to a reduced number of time steps. For the solution expansion two types of basis functions have been proposed, namely FE nodal functions and PWM basis functions. The MPDE approach has been validated on the example of a simplified buck converter. The convergence of the solution in terms of solver tolerance and number of basis functions has been examined. The solution using PWM basis functions converges much faster than when using FE nodal functions. The computational efficiency of the method strongly depends on the choice and number of basis functions. By using the Galerkin approach, the size of the resulting equation system is determined by how many basis functions are used. To solve the final equation system a much smaller number of time steps is necessary however with the drawback of more time spent in each step due to the larger equation systems. A tradeoff between accuracy and speedup is therefore necessary. This becomes particularly visible for the FE nodal functions. When a large number of basis functions are used, the drawback of the approach begins to outweigh the advantage which leads to inefficient simulation. For small number of basis functions the MPDE approach is highly efficient on the presented example of the buck converter.

\section*{Acknowledgements}
This work is supported  by the `Excellence Initiative' of German Federal and State Governments and the Graduate School CE at TU Darmstadt. The authors thank Herbert Egger from the department of mathematics and Felix Wolf from the Graduate School of Computational Engineering at TU Darmstadt for fruitful discussions.

\FloatBarrier
\bibliography{full,english,library,newitems}
\bibliographystyle{ieeetr}
\newpage
\section*{Appendix}
\begin{theorem}
  \label{the:PWMSymmetry}
  The symmetry of the PWM basis functions defined by \eqref{equ:PWMbasis_p1}-\eqref{equ:PWMbasis_normalizing} with duty cycle $D=0.5$ is given by
  \begin{equation*}
    - p_{k}(\tau) = p_{k}(\tau+0.5), \quad \forall \, k=1,\dots,\Np, \text{ and }\, \forall \, \tau \in (0,0.5),
  \end{equation*}
\end{theorem}
\begin{proof}
  The basis functions $p_i(\tau) \; \forall \,i \in \mathbb{N}$ with duty cycle $D=0.5$ are defined as follows:
  The zeroth and first basis function are given piecewisely as
  \begin{equation}
    p_0(\tau)= 1 \; \forall \, \tau \in [0,1],
  \end{equation}
  and
  \begin{align}
    p_1(\tau)= \left\{\begin{array}{cc}
      p_{1,\mathrm{a}}(\tau)=\sqrt{3}(4 \tau - 1), & \forall \,\tau \in [0,0.5)\\
      p_{1,\mathrm{b}}(\tau)=\sqrt{3}(-4 \tau + 3), & \forall \,\tau \in [0.5,1]
    \end{array}\right.,
    \label{equ:PWMBF_1}
  \end{align}
  which essentially corresponds to a scaled and translated hat function.
  The subscript letter refers to the interval in which the basis function is defined, i.e., if ``$\mathrm{a}$'', the polynomial in $\tau\in[0,0.5)$ is considered, if ``$\mathrm{b}$'' the polynomial in $\tau\in[0.5,1]$ is considered. If the subscript comprises only a number, the entire basis function is addressed. 
  
  The symmetry of the basis function $p_1(\tau)$ can be expressed as follows
  \begin{equation}
    \begin{aligned}
      - &p_{1,\mathrm{a}}(\tau) = p_{1,\mathrm{b}}(\tau+0.5) &\quad \, \forall \, \tau\in (0, 0.5) \\
      - &p_{1,\mathrm{a}}(\tau) = p_{1,\mathrm{a}}(0.5-\tau) &\quad \, \forall \, \tau\in (0, 0.5) \\
      &p_{1,\mathrm{a}}(\tau) = p_{1,\mathrm{b}}(1-\tau) &\quad \, \forall \, \tau \in (0, 0.5)
    \end{aligned}\,.
    \label{equ:PWMBF_1_sym}
  \end{equation}
  
  The basis functions of higher order, i.e., $i=2,3,4,\dots$ are calculated by integrating the basis functions of lower order 
  \begin{align}
    p_{i,\mathrm{a}}^\star(\tau) = \int_{0.5}^{\tau} p_{i-1,\mathrm{a}}(\eta) \, \dd \eta \, , \\
    p_{i,\mathrm{b}}^\star(\tau) = \int_{0.5}^{\tau} p_{i-1,\mathrm{b}}(\eta) \, \dd \eta \, , 
  \end{align}
  and orthogonalizing the integrated basis functions against the constant basis function $p_0(\tau)$. The basis function $p_{i,\mathrm{a}}$ therefore is
  \begin{equation}
    \begin{aligned}
      p_{i,\mathrm{a}}(\tau) &= p_{i,\mathrm{a}}^\star(\tau) - \underbrace{\frac{p_0(\tau)}{\sqrt{\int_0^1 p_0(\eta) p_0(\eta) \dd \eta}}}_{=1} \int_0^1 p_{i,\mathrm{a}}^\star(\eta) \underbrace{p_0(\eta)}_{=1} \dd \eta  \\ 
      &= p_{i,\mathrm{a}}^\star(\tau) - \int_0^1 p_{i,\mathrm{a}}^\star(\eta) \dd \eta  \\ 
      &= \int_{0.5}^{\tau} p_{i-1,\mathrm{a}}(\eta) \, \dd \eta - \int_0^1 \int_{0.5}^{\tau} p_{i-1}(\eta) \,\dd \eta \,\dd \tau\\
      &= \int_{0.5}^{\tau} p_{i-1,\mathrm{a}}(\eta) \,\dd \eta - \int_0^{0.5} \int_{0.5}^{\tau} p_{i-1,\mathrm{a}}(\eta) \,\dd \eta \,\dd \tau \\
      &- \int_{0.5}^1 \int_{0.5}^{\tau} p_{i-1,\mathrm{b}}(\eta) \,\dd \eta \,\dd \tau.
    \end{aligned}
    \label{equ:calcPWMBFIntAndOrth}
  \end{equation}
  Similarly $p_{i,\mathrm{b}}$ is given as
  \begin{equation}
    \begin{aligned}
      p_{i,\mathrm{b}}(\tau) &= \int_{0.5}^{\tau} p_{i-1,\mathrm{b}}(\eta) \,\dd \eta - \int_0^{0.5} \int_{0.5}^{\tau} p_{i-1,\mathrm{a}}(\eta) \,\dd \eta \,\dd \tau \\
      &- \int_{0.5}^1 \int_{0.5}^{\tau} p_{i-1,\mathrm{b}}(\eta) \,\dd \eta \,\dd \tau.
    \end{aligned}
    \label{equ:calcPWMBFIntAndOrth2}
  \end{equation}
  Note that orthonormalization of all basis functions against each other is possible and has been originally proposed in \cite{Gyselinck_2013ab}. However, it spans the same space as the basis functions without full orthonormalization \eqref{equ:calcPWMBFIntAndOrth},\eqref{equ:calcPWMBFIntAndOrth2}, and is therefore neglected.
  
  \subsection{Symmetry properties of PWM basis functions}
  The symmetry properties of the basis functions as defined by \eqref{equ:PWMBF_1}-\eqref{equ:calcPWMBFIntAndOrth2} with duty cycle $D=0.5$ are examined in the following.
  
  \subsubsection{Induction hypothesis}
  The symmetry of the basis functions is given by
  \begin{equation}
    \left.\begin{array}{c}
      - p_{i,\mathrm{a}}(\tau) = p_{i,\mathrm{b}}(\tau+0.5) \\
      p_{i,\mathrm{a}}(\tau) = p_{i,\mathrm{a}}(0.5-\tau) \\
      -p_{i,\mathrm{a}}(\tau) = p_{i,\mathrm{b}}(1-\tau) 
    \end{array}\right\}
    \parbox{10em}{$\forall \, i = 2 k, k\in\mathbb{N} \setminus \{0\} \\\text{ and }\, \forall \, \tau \in (0,0.5)$},
    \label{equ:hypo_sym_evenIndex}
  \end{equation}
  i.e., for all basis functions with even index, and 
  \begin{equation}
    \left.\begin{array}{c}
      - p_{i,\mathrm{a}}(\tau) = p_{i,\mathrm{b}}(\tau+0.5) \\
      - p_{i,\mathrm{a}}(\tau) = p_{i,\mathrm{a}}(0.5-\tau) \\
      p_{i,\mathrm{a}}(\tau) = p_{i,\mathrm{b}}(1-\tau) 
    \end{array}\right\}
    \parbox{10em}{$\forall \, i = 1+2 k, k\in\mathbb{N} \setminus \{0\}\\\text{ and }\, \forall \, \tau \in (0,0.5)$},
    \label{equ:hypo_sym_oddIndex}
  \end{equation}
  i.e., for all basis functions with odd index.
  
  \subsubsection{Induction base}
  We calculate the basis functions $p_2(\tau)$ and $p_3(\tau)$ and their symmetry properties. 
  The basis function $p_2(\tau)$ is obtained using \eqref{equ:calcPWMBFIntAndOrth},\eqref{equ:calcPWMBFIntAndOrth2} and given by
  \begin{equation}
    p_{2,\mathrm{a}}(\tau)=\sqrt{3}\,\left(\,2 \tau^2-\tau\,\right),
  \end{equation}
  and
  \begin{equation}
    p_{2,\mathrm{b}}(\tau)=\sqrt{3}\,\left(\,2\tau^2+3\tau-1\,\right).
  \end{equation}
  They fulfill the symmetry properties stated in \eqref{equ:hypo_sym_evenIndex}.
  
  The basis function $p_3(\tau)$ is also obtained using \eqref{equ:calcPWMBFIntAndOrth}, \eqref{equ:calcPWMBFIntAndOrth2} and given by
  \begin{equation}
    p_{3,\mathrm{a}}(\tau)=\sqrt{3}\,\left(\,\frac{2}{3}\tau^3 - 0.5 \tau^2 + \frac{1}{48}\,\right),
  \end{equation}
  and
  \begin{equation}
    p_{3,\mathrm{b}}(\tau)=\sqrt{3}\,\left(\,-\frac{2}{3}\tau^3 + \frac{3}{2}\tau^2 - \tau + \frac{3}{16}\,\right).
  \end{equation}
  They fulfill the symmetry properties stated in \eqref{equ:hypo_sym_oddIndex}.
  
  \subsubsection{Induction step}
  We calculate the basis functions $p_i(\tau)$ and $p_{i+1}(\tau)$, where $i=2k, k\in\mathbb{N} \setminus \{0,1\}$ and their symmetry properties. 
  \paragraph{Basis function $p_i(\tau)$.}
  The basis function $p_i(\tau)$ is given by integration and orthogonalization against the constant basis function $p_0(\tau)$, i.e., 
  \begin{equation}
    p_{i,\mathrm{a}}(\tau) = \int_{0.5}^{\tau} p_{i-1,\mathrm{a}}(\eta) \,\dd \eta - \int_0^1 \int_{0.5}^{\tau} p_{i-1}(\eta) \,\dd \eta \,\dd \tau
    \label{equ:PWMBF_ia}
  \end{equation}
  and
  \begin{equation}
    p_{i,\mathrm{b}}(\tau) = \int_{0.5}^{\tau} p_{i-1,\mathrm{b}}(\eta) \,\dd \eta - \int_0^1 \int_{0.5}^{\tau} p_{i-1}(\eta) \,\dd \eta \,\dd \tau
    \label{equ:PWMBF_ib}
  \end{equation}
  The orthogonalization term yields using the symmetry properties \eqref{equ:hypo_sym_oddIndex} and substitution
  \begin{equation}
    \begin{aligned}
      \int_0^1 \int_{0.5}^{\tau} p_{i-1}(\eta) \, \dd \eta \, \dd \tau &= \int_0^{0.5} \int_{0.5}^{\tau} p_{i-1,\mathrm{a}}(\eta) \, \dd \eta \, \dd \tau \\
      &+ \int_{0.5}^{0} \int_{0.5}^{\tau} p_{i-1,\mathrm{a}}(\eta) \, \dd \eta \, \dd \tau = 0.
    \end{aligned}
  \end{equation}
  Therefore the expressions for the basis function $p_i(\tau)$ \eqref{equ:PWMBF_ia} and \eqref{equ:PWMBF_ib} simplify to
  \begin{equation}
    p_{i,\mathrm{a}}(\tau) = \int_{0.5}^{\tau} p_{i-1,\mathrm{a}}(\eta) \, \dd \eta
  \end{equation}
  and
  \begin{equation}
    p_{i,\mathrm{b}}(\tau) = \int_{0.5}^{\tau} p_{i-1,\mathrm{b}}(\eta) \, \dd \eta.
  \end{equation}
  The symmetries are given by, using \eqref{equ:hypo_sym_oddIndex},
  \begin{equation}
    \begin{aligned}
      -p_{i,\mathrm{a}}(\tau)&=p_{i,\mathrm{b}}(\tau+0.5),
    \end{aligned}
    \label{equ:PWMBF_i_sym1}
  \end{equation}
  and
  \begin{equation}
    \begin{aligned}
      p_{i,\mathrm{a}}(\tau)&= p_{i,\mathrm{a}}(0.5-\tau),
    \end{aligned}
    \label{equ:PWMBF_i_sym2}
  \end{equation}
  and thus fulfill the hypothesis \eqref{equ:hypo_sym_evenIndex}.
  
  \paragraph{Basis function $p_{i+1}(\tau)$.}
  The basis function $p_{i+1}(\tau)$ is calculated by
  \begin{equation}
    p_{i+1,\mathrm{a}}(\tau) = \int_{0.5}^{\tau} p_{i,\mathrm{a}}(\eta) \, \dd \eta - \int_0^1 \int_{0.5}^{\tau} p_{i}(\eta) \, \dd \eta \, \dd \tau,
    \label{equ:PWMBF_i+1a}
  \end{equation}
  and
  \begin{equation}
    p_{i+1,\mathrm{b}}(\tau) = \int_{0.5}^{\tau} p_{i,\mathrm{b}}(\eta) \, \dd \eta - \int_0^1 \int_{0.5}^{\tau} p_{i}(\eta) \, \dd \eta \, \dd \tau.
    \label{equ:PWMBF_i+1b}
  \end{equation}
  The orthogonalization term is calculated using the symmetry properties \eqref{equ:hypo_sym_evenIndex} and substitution
  \begin{align*}
    \int_0^1 \int_{0.5}^{\tau} p_{i}(\eta) \, \dd \eta \, \dd \tau &=  - 0.5 \int_{0}^{0.5} p_{i,\mathrm{a}}(\eta) \,\dd \eta.
  \end{align*}
  Therefore the expressions for the basis function $p_{i+1}(\tau)$ \eqref{equ:PWMBF_i+1a} and \eqref{equ:PWMBF_i+1b} are
  \begin{equation}
    p_{i+1,\mathrm{a}}(\tau) = \int_{0.5}^{\tau} p_{i,\mathrm{a}}(\eta) \,\dd \eta + 0.5 \int_{0}^{0.5} p_{i,\mathrm{a}}(\eta) \,\dd \eta,
  \end{equation}
  and
  \begin{equation}
    p_{i+1,\mathrm{b}}(\tau) = \int_{0.5}^{\tau} p_{i,\mathrm{b}}(\eta) \,\dd \eta + 0.5 \int_{0}^{0.5} p_{i,\mathrm{a}}(\eta) \,\dd \eta.
  \end{equation}
  The symmetries are given by, using \eqref{equ:hypo_sym_evenIndex},
  \begin{equation}
    \begin{aligned}
      -p_{i+1,\mathrm{a}}(\tau)&=p_{i+1,\mathrm{b}}(\tau+0.5),
    \end{aligned}
    \label{equ:PWMBF_i+1_sym1}
  \end{equation}
  and
  \begin{equation}
    \begin{aligned}
      -p_{i+1,\mathrm{a}}(\tau)&=p_{i+1,\mathrm{a}}(0.5-\tau),
    \end{aligned}
    \label{equ:PWMBF_i+1_sym2}
  \end{equation}
  and thus fulfill the hypothesis \eqref{equ:hypo_sym_oddIndex}.
\end{proof}

\begin{remark}
  The PWM basis functions are suited to approximate the solution of linear ODEs with 2-level pulsed excitation. The solution of these ODEs are given by piecewise exponential functions which fulfill the symmetry condition stated in Theorem \ref{the:PWMSymmetry}.
\end{remark}
We calculate the solution of the linear ODE
\begin{equation}
  A \ddt x(t) + B x(t) = c(t)
\end{equation}
where $A, B \in \mathbb{R}$ are constants, $x(t)\in\mathbb{R}$ is the solution and $c(t)\in\mathbb{R}$ is the excitation. Generalization to systems of ODEs is straightforward. Without loss of generality we assume $A=1$. 
Rewriting leads to
\begin{equation}
  \ddt x(t) = c(t) - B x(t).
  \label{equ:buckODE}
\end{equation}
Therefore the homogeneous problem is given as 
\begin{equation}
  \ddt x(t) = - B x(t).
\end{equation}
The solution of the ODE \eqref{equ:buckODE} is given by
\begin{equation}
  x(t)=\alpha\e^{- B t} + x_p(t),
\end{equation}
where $\alpha\in\mathbb{R}$ is a constant and $x_p(t)$ is a particular solution.
In the following the time interval of one period of a 2-level pulsed excitation $c(t)$ with duty cycle $D=0.5$ is considered. The excitation is given by
\begin{equation}
  c(t)=\left\{
  \begin{array}{ll}
    1 &\text{for } 0\leq t < 0.5\,\Ts \\
    -1 &\text{for } 0.5\,\Ts \leq t \leq \Ts
  \end{array}\right..
\end{equation}
Two cases are distinguished. Either $0\leq t < 0.5\,\Ts$ or $0.5\,\Ts \leq t \leq \Ts$. In the first case, the solution and constants are denoted with additional subscript ``a'', in the second case with additional subscript ``b''.

The solution for the first interval is then given by 
\begin{equation}
  \begin{aligned}
    x_\rma(t)=\alpha_\rma \e^{-B t} + B^{-1},
  \end{aligned}
\end{equation}
where the last term is a particular solution if the excitation is constantly $1$.

The solution for the second interval is given by 
\begin{equation}
  \begin{aligned}
    x_\rmb(t)=\alpha_\rmb \e^{-B t} - B^{-1},
  \end{aligned}
\end{equation}
where the last term is a particular solution if the excitation is constantly $-1$.

The following conditions require to be satisfied for each ripple of the solution:
\begin{align}
  x_\rma(0)&=x_\rmb(\Ts) \\
  x_\rma(0.5\,\Ts)&=x_\rmb(0.5\,\Ts).
\end{align}

Inserting these conditions into the solutions gives the two equations
\begin{align}
  \alpha_\rma \e^{- B 0} + B^{-1} &= \alpha_\rmb \e^{- B \Ts} - B^{-1} \\
  \alpha_\rma \e^{- B 0.5 \Ts} + B^{-1} &= \alpha_\rmb \e^{- B 0.5 \Ts} - B^{-1}
\end{align}
Substracting the second from the first equation leads to the relation
\begin{equation}
  \alpha_\rma = -\alpha_\rmb \e^{- B 0.5 \Ts}
\end{equation}
The symmetry of the solution is, using the relation between the coefficients, given by 
\begin{equation}
  \begin{aligned}
    x_\rma(t-0.5 \Ts)&=\alpha_\rma \e^{- B t} \e^{ B 0.5 \Ts} + B^{-1} \\
    &=-\alpha_\rmb \e^{- B 0.5 \Ts} \e^{- B t} \e^{ B 0.5 \Ts} + B^{-1} \\
    &=-\alpha_\rmb \e^{- B t} + B^{-1} \\
    &=-x_\rmb(t).
  \end{aligned}
  \label{equ:symmetryODESolution}
\end{equation}
The PWM basis functions are polynomials of degree up to $\Np$, which span the polynomial space of dimension $\Np + 1$. Thus, a linear combination of them with duty cycle $D=0.5$ can exactly represent any piecewise polynomial with $C^0$ continuity at $\tau=0.5$, maximum degree $\Np$ and symmetry condition common to odd and even indexed PWM basis functions \eqref{equ:hypo_sym_evenIndex}, \eqref{equ:hypo_sym_oddIndex}, i.e., $- p_{i}(\tau) = p_{i}(\tau+0.5)$. The solution of the linear ODE fulfills this condition, see \eqref{equ:symmetryODESolution}.

\begin{remark}
  The PWM basis functions may not be suited to represent the solution of linear or nonlinear ODEs with arbitrary excitations. A counter example is an ODE with a 3-level pulsed excitation (see Fig. \ref{fig:excitationProjection}), for which it can be shown that the approximation fails.
\end{remark}

\begin{figure}
  \centering
  \includegraphics{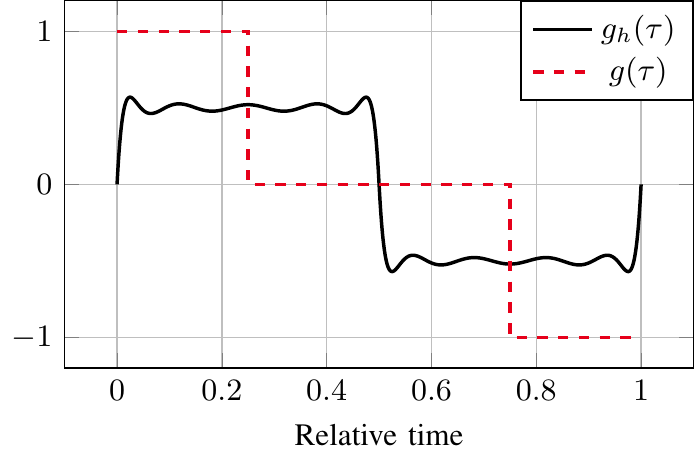}
  \caption{3-level pulsed excitation $g(\tau)$ and its projection $g_h(\tau)$ onto the space spanned by $p_k(\tau)$ with $\Np=10$.}
  \label{fig:excitationProjection}
\end{figure}

The PWM basis functions above are built for a particular duty cycle $D$ to represent piecewise exponential solutions as generated in power converters by 2-level pulsed excitations. Let us show that they do not span $\Ltwo([0,1])$. Consider as example the 3-level function
\begin{equation}
  g(\tau) \ = \ \left\{  
  \begin{array}{ll} 
    1 \quad & \text{for} \ \ \ 0 \leq \tau \leq 0.25 \\ 
    0 \quad & \text{for} \ \ \ 0.25 \leq \tau \leq 0.75 \\
    -1 \quad & \text{for} \ \ \ 0.75 \leq \tau \leq 1
  \end{array}
  \right. .
  \label{equ:functiong}
\end{equation}
It is depicted in Fig.~\ref{fig:excitationProjection}. 

$g(\tau)$ is $\Ltwo$-projected onto the space spanned by the basis functions $p_k(\tau) \, \forall k \in \mathbb{N}$. The projection $g_h(\tau)$ (see Fig.~\ref{fig:excitationProjection}) is a linear combination of the basis functions 
\begin{equation}
  g_h(\tau)=a_0 \, p_0(\tau) + a_1 \, p_1(\tau) + \dots + a_{\Np} \, p_{\Np}(\tau),
\end{equation}
where $\Np$ is the number of employed basis functions. 
The zeroth basis function and all basis functions with odd index do not contribute to $g_h(\tau)$ as
\begin{equation}
  \int_{0}^{1} p_k(\tau) \, g(\tau) \, \dd \tau = 0,\quad k=0,1,3,5,\dots.
\end{equation}
Therefore the final solution exhibits the same symmetry properties as the basis functions with even index. These are given by \eqref{equ:hypo_sym_evenIndex}. We assume without loss of generality, that $g_h(\tau)$ is given in terms of orthonormalized basis functions. 
As they span the same space, the symmetry properties of $g_h(\tau)$ do not change. The error between $g_h(\tau)$ and $g(\tau)$ in the $\Ltwo$ sense can be estimated as follows, where, for simplicity, the $\tau$ dependency is omitted 
\begin{align}
  \int_{0}^{1} (g_h-g)^2 \, \dd \tau 
  \overset{\eqref{equ:hypo_sym_evenIndex}}{=}& 2 \int_{0}^{0.5} (g_h-g)^2 \, \dd \tau \\
  = & 2 \int_{0}^{0.5} g^2-2 g g_h + g_h^2 \, \dd \tau \\
  = & 2 \int_{0}^{0.25} g^2-2 g g_h + g_h^2 \, \dd \tau \\
  & + 2 \int_{0.25}^{0.5} g^2-2 g g_h + g_h^2 \, \dd \tau.\\
  \intertext{Using $g(\tau)=0 \,\, \forall \tau \in [0.25,0.5]$ yields}
  \int_{0}^{1} (g_h-g)^2 \, \dd \tau= & 2 \int_{0}^{0.25} g^2-2 g g_h + g_h^2 \, \dd \tau + 2 \int_{0.25}^{0.5} g_h^2 \, \dd \tau \\
  \overset{\eqref{equ:hypo_sym_evenIndex}}{=} & 2 \int_{0}^{0.25} g^2-2 g g_h + g_h^2 + g_h^2 \, \dd \tau \\
  = &2 \norm{g-g_h}^2_{\Ltwo([0,0.25])}\\
  &+ 2 \norm{g_h}^2_{\Ltwo([0,0.25])}.
\end{align}
Now the expression $\norm{g_h}_{\Ltwo([0,0.25])}$ is estimated using the orthonormality of the basis functions. The $\Ltwo$ scalar product is denoted as $\langle a(\tau),b(\tau) \rangle = \int_{0}^{0.25} a(\tau) \, b(\tau) \, \dd \tau$, where in the following we leave out the $\tau$ dependency for simplicity
\begin{align}
  \norm{g_h}^2_{\Ltwo([0,0.25])} &= \langle \sum_k \langle g,p_k \rangle p_k, \sum_l \langle g,p_l \rangle p_l \rangle \\
  &= \sum_k \sum_l \langle \langle g,p_k \rangle p_k, \langle g,p_l \rangle p_l \rangle \\
  \intertext{Using the orthonormality of the basis yields}
  \norm{g_h}^2_{\Ltwo([0,0.25])} &=\sum_k \langle \langle g,p_k \rangle p_k, \langle g,p_k \rangle p_k \rangle \\
  &=\sum_k \langle \langle g,p_k \rangle, \langle g,p_k \rangle \rangle \\
  &=\sum_k \langle g,p_k \rangle^2 \\
\end{align}
As $\langle g,p_k \rangle^2$ is always positive independent of how many basis functions are used, the error $\norm{g_h-g}^2_{\Ltwo([0,1])}$ will always be greater than a fixed constant. 

\end{document}